\documentclass[10pt,conference]{IEEEtran}

\usepackage[english]{babel}
\usepackage{amssymb}
\usepackage{amsmath}
\usepackage{amsthm}
\usepackage{graphicx}
\usepackage{xspace}
\usepackage{subcaption}
\usepackage{braket}
\usepackage[colorlinks=true, allcolors=blue]{hyperref}

\newcommand{\om}{group-evolve\xspace}

\newtheorem{theorem}{Theorem}

\newtheorem{obs}[theorem]{Observation}

\title{An Analysis of Commutation-Based Trotter Ordering Strategies on Heisenberg-Style Hamiltonians}

\author{\IEEEauthorblockN{Reuben Tate, Shamminuj Aktar, Stephan Eidenbenz}
\IEEEauthorblockA{\textit{CAI-3: Information Sciences} \\
\textit{Los Alamos National Laboratory}\\
Los Alamos, New Mexico, USA}
}
\begin{document}
\maketitle

\begin{abstract}
Trotterization is a technique that allows one to approximate a time evolution of a Hamiltonian by repeatedly evolving the individual terms of the Hamiltonian one-at-a-time for small time durations. Bounds on the error of this approximation exist; however, they are typically loose and moreover, it is known that the true error can be greatly influenced by the order in which the terms of the Hamiltonian are evolved. In this work, we consider various ordering strategies that exploit the commutation structure of the Hamiltonian, in addition to a few other baseline ordering strategies. These commutation-based strategies involve dividing the terms of the Hamiltonian into groups where all the terms within each group commute with one another. These groupings can be obtained by using graph coloring techniques on what we call the ``commutation graph" of the Hamiltonian. We prove various results regarding the structure and properties of such commutation graphs for certain classes of Hamiltonians. We also empirically calculate the (true) Trotter error using these ordering strategies on various 1D and 2D Heisenberg-style systems.
\end{abstract}

\section{Introduction}
Trotterization, also known as the Lie–Trotter–Suzuki product formula \cite{trotter1959product, suzuki1990fractal, suzuki1992general}, is one of the most widely used techniques for simulating quantum dynamics. Given a Hamiltonian decomposed as a sum of non-commuting terms $H = \sum_j H_j$, the exact time-evolution operator $e^{-iHt}$, can be approximated by a product of exponentials of the individual terms, applied over many short time steps \cite{lloyd1996universal}. This decomposition enables a direct mapping of continuous-time quantum evolution onto discrete gate sequences, making it especially attractive for digital quantum simulation on gate-based quantum computers.

The usefulness of Trotterization stems from both its conceptual simplicity and its flexibility. By increasing the number of Trotter steps (or equivalently decreasing the step size), the approximation error can in principle be made arbitrarily small, and higher-order Suzuki formulas can further improve accuracy. Consequently, Trotter-based methods have become a standard tool for Hamiltonian simulation, with applications ranging from quantum chemistry to condensed matter physics and quantum algorithms such as quantum phase estimation (QPE) \cite{lloyd1996universal}. At the same time, there is an inherent trade-off between approximation accuracy and resource cost: increasing the number of steps improves the mathematical approximation but leads to longer circuits and greater exposure to hardware noise.

A substantial theoretical literature has established rigorous bounds on the error introduced by Trotterization, typically expressed in terms of norms of nested commutators and scaling as powers of the timestep and simulation time \cite{childs2021theory}. However, it is now well understood that these bounds can be highly pessimistic. In quantum chemistry, for example, worst-case bounds on Trotter error have been shown to overestimate the true error by many orders of magnitude \cite{babbush2015chemical,reiher2017elucidating}. More generally, standard analyses often neglect structure such as locality, commutativity, or entanglement, leading to bounds that fail to capture typical performance. 

Due to non-commutativity, the exact Trotter error can change by simply rearranging the ordering of the terms in the Hamiltonian; most Trotter error bounds are order-agnostic and thus cannot capture the effects of re-ordering. In this work, we examine the effect of various ordering strategies and the effect they have on the (exact) Trotter error.

The existence of Trotter error is due to the inherent non-commutativity between certain pairs of terms. This motivates an approach that considers the commutation structure between all pairs of terms in the Hamiltonian and we use this structure to find groups of terms where all the terms within each group commute with one another; we formalize this approach via a graph-theoretic lens. In prior work \cite{tranter2019ordering}, the Trotter ordering was determined by picking out terms, one-by-one, among such groups based on some simple rule; they referred to these approaches as equaliseGroups and depleteGroups. We consider what is arguably a more natural approach where we time-evolve entire groups at a time; since each group has terms that mutually commute, the evolution of each group is exact. We call this method the \emph{\om} method. We compare these approaches against baseline approaches where the terms are ordered either by the magnitude of the coefficients or by a lexicographical ordering.

Our results confirm previous findings that show that Trotter error is heavily dependent on the ordering \cite{tranter2019ordering}. While the relative performance differences across different ordering strategies is dependent on several factors (Trotter order, number of Trotter steps, etc), we find that our \om method is often the top-performing one.

\section{Background}

\subsection{Hamiltonian Basics}
\label{sec:hamBasics}
We let $X,Y,Z$ denote the Pauli matrices. We let $X_i, Y_i, Z_i$ denote the Pauli $X,Y,Z$ matrix acting on the $i$th qubit respectively. A \emph{Pauli string} is any tensor product of Pauli and $2\times 2$ identity matrices. 
%A Pauli string is considered $k$-local if the number of Pauli matrices that appear as factors in the corresponding tensor product is at most $k$; we say that it is strictly $k$-local if the number such Pauli matrices is exactly $k$.
A Pauli string is considered non-mixed if it only contains at most one type of Pauli matrix among $\{X,Y,Z\}$, e.g., $I$, $X_2X_4$, $Y_3$, $Z_1Z_2Z_4$ are examples of non-mixed Pauli strings. 

Any $n$-qubit Hamiltonian can be written uniquely as a linear combination of Pauli strings; many of the results in this paper assume one already has the Hamiltonian already represented in such a form. If all of the Pauli strings of a Hamiltonian are non-mixed, we say that the Hamiltonian itself is non-mixed.

\subsection{Trotter Ordering and Trotter Error}
For a Hamiltonian $H = \sum_{j=1}^k H_j$, the 1st-order Trotterization with time $t$ and $s$ steps is given by:
$$\left(\prod_{j=1}^k \exp\left(-i\frac{t}{s} H_j\right)\right)^s \approx \exp(-itH).$$

There are various notions of Trotter error. Using $\mathcal{F}(t)$ to denote the unitary corresponding to some Trotterization, one of the most common Trotter error types is to consider the matrix norm of the difference of the unitaries, i.e. $\Vert e^{-itH} - \mathcal{F}(t)\Vert,$ for some choice of norm.

In this work, we consider a state-dependent notion of Trotter error which is computationally faster to compute (see Section \ref{sec:eval_method}). In particular, for a given initial state $\ket{\psi_i}$, we calculate the fidelity which is measured as
\begin{equation}
  F = \bigl|\langle \psi_{\mathrm{exact}} \mid
             \psi_{\mathrm{Trotter}} \rangle\bigr|^2,
  \label{eq:fidelity}
\end{equation}
where $|\psi_{\mathrm{exact}}\rangle = e^{-iHt}\ket{\psi_i}$ and $|\psi_{\mathrm{Trotter}}\rangle = \mathcal{F}(t)\ket{\psi_i}$~\cite{jozsa1994fidelity}. See Section \ref{sec:eval_method} regarding our choice of initial state in our empirical results.

Formally, a Trotter ordering can be associated with a permutation function $\sigma: [k] \to [k]$, in which case, the Trotterization becomes:
$$\exp(-itH) \approx \left(\prod_{j=1}^k \exp\left(-i\frac{t}{s} H_{\sigma(j)}\right)\right)^s.$$

For most Hamiltonians, the optimal choice of ordering $\sigma$ that minimizes the Trotter error becomes intractable as the number of terms, $k$, increases; a brute force search requires the consideration of $k!$ different orderings. We investigate different strategies for constructing good orderings in the following subsections.

To avoid confusion, we want to remind the reader that the Trotter \emph{ordering} and the Trotter \emph{order} are two different concepts. The former is the primary focus of this paper; meanwhile, the Trotter order $p$ refers to categories of product formulas whose Trotter error (in the traditional sense) scales as $O(t^{p+1})$ for a fixed number of steps $r$ \cite{suzuki1990fractal, suzuki1992general}.

For our \om strategy, we consider \emph{groups} of terms where the terms within each group commute with one another. For 1st order and 2nd order Trotterization, these groupings can be ``flattened" into an ordering; the ordering \emph{within} each group does not matter in this case. For higher-order Trotter formulas however; it is advised that each group be treated as a single ``unit" for the purposes of Trotterization; we explain this more in Section \ref{sec:grouping_ham_terms}.

\section{Ordering Strategies}

\subsection{Commutation Graphs}
\label{sec:commutationGraphs}
The primary ordering strategies we consider are based on the commutativity structure of the Hamiltonian. We formalize this by creating the \emph{commutation graph} for each Hamiltonian which we define below.

We denote the commutation graph of a Hamiltonian
%\footnote{There is some abuse of notation, as the commutation graph is formally defined on some collection of sub-Hamiltonians (Pauli strings in our case). Since, in this work, we are only dealing with Pauli string decompositions, which are unique, then this will not be an issue.} 
$H$ as $C(H)$. Let $H = \sum_{i \in I} c_i P_i$ be the decomposition of $H$ into Pauli strings. Then the vertex set of $C(H)$ is exactly the set of Pauli strings $\{P_i\}_{i \in [k]}$, and two Pauli strings $P$ and $P'$ are considered adjacent in $C(H)$ if $P$ and $P'$ do \emph{not} commute. 
% More succintly:
% $$V(C(H)) = \{P_i\}_{i \in I}$$
% $$E(C(H)) = \left\{(P, P') \in {V \choose 2} :[P, P'] \neq 0\right\}.$$

By construction, any independent set in $C(H)$ (i.e. a collection of vertices with no edges between them in $C(H)$), corresponds to a group of Pauli terms that all commute with one another. Since a proper (vertex) coloring of any graph partitions the vertices into disjoint collections of independent sets (based on the color classes), we have the following observation.

\begin{obs}
    Let $c$ be a coloring of $C(H)$ with the minimum $\chi(C(H))$ number of colors; then the color classes of $c$ correspond to a disjoint collection of groups of Pauli strings where the Pauli strings within each group commutes and the number of groups is as small as possible.
\end{obs}

In other words, finding a valid grouping of the Pauli strings (as defined in Section \ref{sec:grouping_ham_terms}) with as few groups as possible is equivalent to vertex-coloring some graph with as few colors as possible. Although the problem of finding such an optimal coloring for arbitrary graphs is known to be NP-Hard \cite{karp2009reducibility}, we will later see (Section \ref{sec:theory}) that for certain categories of Hamiltonians $H$, that the commutation graph $C(H)$ has certain properties and structures that can be exploited, making them easy to color.

% For each Hamiltonian $H$, we will consider what we call its commutation graph $G(H)$. Each vertex of $G(H)$ corresponds to a term in the Hamiltonian. Note that the notion of a ``term" will be dependent on how the Hamiltonian is expressed; when the context is not clear, we will add a subscript (i.e. JW for Jordan-Wigner, BK for Bravyi-Kitaev, and Fer for the original pre-transformed formulation with fermionic terms).

% For each Hamiltonian term $h_1, h_2 \in V(G)$

The runtime to construct the commutation graph scales as $O(k^2 f(n))$ where $k$ is the number of terms and $f(n)$ is the time it takes to check commutation of any pair of terms. It is known that for Pauli strings, a naive approach gives $f(n) = O(n)$; however, if the Pauli strings are represented in a particular way (bitstring symplectic form), and if the representation fits in a constant number of machine words, then commutativity checking is (effectively) $O(1)$ in the word-RAM model of computing \cite{dion2024efficiently, gottesman1997stabilizer}.

\subsection{Commutation-Based Ordering Strategies}
We consider 3 different ordering strategies that utilize the coloring/groups of the commutation graph described in the previous section.

We first describe our method, which we call \emph{\om}. For each group, we (exactly) evolve the entire group at once. For 1st-order Trotterization, this is equivalent to picking an ordering $\sigma$ where we first pick all the terms in the first group, then all the terms in the second group, etc. We make this approach more formal in the next section.

The other two methods, \emph{equaliseGroups} and \emph{depleteGroups}, also use the groups from the commutation graph; we refer the reader to \cite{tranter2019ordering} for more details. We break ties by index in the original file. 

% \textbf{equaliseGroups}: performs the equaliseGroups strategy in \cite{tranter2019ordering}. Builds the ordering by picking out the term with the highest magnitude coefficient within the group that has the most terms. If there are multiple groups with the same number of terms, we pick out the term with the highest magnitude among the union of those groups. For terms with the same magnitude, \cite{tranter2019ordering} does not discuss how ties are broken, but we break ties by index in the original file. 

% \textbf{depleteGroups}: performs the depleteGroups strategy in \cite{tranter2019ordering}. Cycles through groups, picking the term with the highest coefficient magnitude amongst remaining terms, and appends it to the ordering.  For terms with the same magnitude,\cite{tranter2019ordering} does not discuss how ties are broken, but we break tiesby index in the original file.

\subsection{Grouping Hamiltonian Terms}
\label{sec:grouping_ham_terms}
Suppose $H$ is expressed as a sum of Pauli strings
%\footnote{The key idea in this section still works if the $P_i$'s are replaced by other sub-Hamiltonians $H_i$ for which we know how to execute $\exp(-itH_i)$ exactly on a quantum device; we assume these are Pauli strings in this work for ease of presentation.}
(e.g. $H = \sum_{i \in I} P_i$ with some index set $I$); for simplicity, here we assume that the coefficient of the Pauli string is absorbed into each $P_i$. Let us consider \emph{partitioning} the set of Pauli terms into $k$ disjoint groups $P^{(1)}, P^{(2)}, \dots,  P^{(k)}$, where the Pauli terms within each group all commute with one another, i.e., for all $i$, if $P, P' \in P^{(i)}$, then $P$ and $P'$ commute. For each Pauli grouping, $P^{(i)}$, let $H^{(i)}$ be the corresponding Hamiltonian, i.e., $H^{(i)} = \sum_{P \in  P^{(i)}} P$. Observe that $H = \sum_{i=1}^k H^{(i)}$.

Instead of performing Trotterization on the individual Pauli strings, we consider Trotterization with respect to the groups, i.e., we approximate the time-evolution (with time $t$) of Hamiltonian $H$ as:
\begin{equation}\label{eqn:approx_ham_group}\exp(-iHt) \approx \left(\prod_{i \in [k]} \exp(-itH^{(i)}/s)\right)^s ,\end{equation}
for some integer number of steps $s\geq 1$. Since all the Pauli terms commute within each group, we have the familiar exponential rule for each group:

\begin{equation}\label{eqn:exact_time_evolution}\exp(-itH^{(i)}/s) = \prod_{P \in  P^{(i)}} \exp(-itP/s),\end{equation}
% \begin{equation}\label{eqn:exact_time_evolution}\exp(-itH^{(i)}/s) = \exp\left(\sum_{P \in  P^{(i)}} -itP/s\right),\end{equation}
% $$= \prod_{P \in  P^{(i)}} \exp(-itP/s),$$

meaning that the time-evolution of $H^{(i)}$ can be executed \emph{exactly} through a sequence of operations that are easily implementable on a quantum device.  This means that the Trotter error \emph{within} each group is 0. Note that the order of the Pauli terms in the product does not matter (if $H$ and $H'$ commute, then it is straightforward to show that $\exp(-itH)$ and $\exp(-itH')$ commute for any $t$). Combining the above two equations yields the following approximation for the full Hamiltonian:
\begin{equation}\label{eqn:approx_ham_group_2}\exp(-iHt) \approx \left(\prod_{i\in [k]} \left(\prod_{P \in P^{(i)}} \exp(-itP/s)\right)\right)^s,\end{equation}

which is effectively a Trotterization where the ordering of the Pauli strings is (partly) determined by the groups $\{P^{(i)}\}_{i=1}^k$. Note that the ordering \emph{within} each group (inner product of Equation \ref{eqn:approx_ham_group_2}) does not matter due to commutativity; however, the ordering of the groups themselves (outer product of Equation \ref{eqn:approx_ham_group_2}) will yield different results (and thus potentially different Trotter errors).

In the optimistic case where there are $k=1$ groups, this means that the Trotterization of the entire Hamiltonian $H$ has zero error. In the trivial case where there are $k = |I|$ groups, i.e., each group has just a single Pauli string, the approximation in Equation \ref{eqn:approx_ham_group_2} is effectively just a ``regular" Trotterization, i.e., nothing is gained. Intuitively, larger groups are more ideal since that allows us to perform exact time evolution for a large chunk of the Hamiltonian exactly. Larger (average) group sizes correspond to a smaller number of groups; with $k=1$ groups being the most ideal case. For this reason, we use the number of groups $k$, as a proxy for the quality of the grouping  $\{P^{(i)}\}_{i \in [k]}$. Thus, optimizing via this proxy corresponds to minimizing the number of groups $k$, which, as we saw before (Section \ref{sec:commutationGraphs}), is exactly equivalent to finding a minimum vertex coloring of the Hamiltonian's commutation graph.

\subsection{Baseline Ordering Strategies}
\label{sec:baseline_orderings}

Below, we present two simple baseline ordering strategies (introduced in \cite{tranter2019ordering}) to compare our commutation-based ordering strategies against.

\textbf{Magnitude Ordering:} In the magnitude ordering scheme, we sort the Hamiltonian terms in descending order based on the absolute value of the coefficients associated with that term, as is done in \cite{tranter2019ordering}. We break ties based on the index of the Pauli string in the file.

\textbf{Lexicographical Ordering:} We write each term as an $n$-character Pauli string with identities, e.g., for a 5-qubit system, we write $Z_1X_2Y_4$ as $ZXIYI$. We then sort these $n$-character strings lexicographically with the convention $I < X < Y < Z$. This ordering is heavily dependent on how the qubits are indexed; we describe the indexing schemes of the Hamiltonians we consider in Appendix \ref{sec:qubitIndexing}. We did consider a lexicographical ordering on the ``sparse" representation of the Pauli strings; however, preliminary results had shown that such orderings are typically worse than random, and thus, we do not include such sparse lexicographical orderings in this work.

% \textbf{Sparse Lexicographical Ordering:} In \cite{tranter2019ordering}, they also consider a lexicographical ordering of the Hamiltonian terms. Such an ordering is heavily dependent on the representation of the each Hamiltonian term and also the qubit indexing scheme; unfortunately, the details of the lexicographical ordering of \cite{tranter2019ordering} are unclear. Given a Hamiltonian written as a sum of Pauli strings, we assume each Pauli string is written in a ``sparse" format such as $X_2Y_3X_5Z_6$ (instead of a dense format as $IXYIXZ$). Each sparsely written pauli string is then converted into a pair of tuples where the first tuple contains the indices and the second tuple contains the type of Pauli, i.e., $X_2Y_3X_5Z_6$ becomes $((2,3,5,6), (X, Y, X, Z))$. The lexicographical ordering is then performed on such pairs of tuples; when comparing Pauli types, we use the convention that $X < Y < Z$. The effect of doing the lexicographical ordering in this manner ensures that pauli strings with the same set of ``active" indices appear next to each other in the ordering. We describe the indexing schemes of the Hamiltonians we consider in Appendix \ref{sec:qubitIndexing}.

\section{Theory}
\label{sec:theory}
\subsection{Bounds on Number of Colors}
Recall from Section \ref{sec:hamBasics} that a non-mixed Hamiltonian has only Pauli strings where each string contains at most one Pauli type. For such Hamiltonians, we can explicitly find a coloring, which we call the \emph{XYZ-coloring}, with at most 3 colors, one for each Pauli type.
\begin{theorem}[XYZ-Coloring]
\label{thm:xyz_coloring}
    Let $H = \sum_{i \in I} P_i$ be a non-mixed Hamiltonian. Then $\chi(C(H)) \leq 3$.
\end{theorem}
\begin{proof}
    Since $H$ is non-mixed, the set of Pauli strings can be decomposed as a disjoint union $\{P_i\}_{i\in I} = P^X \cup P^Y \cup P^Z$ where for $\alpha \in \{X,Y,Z\}$, $P^\alpha$ denotes the set of Pauli strings in $\{P_i\}_{i\in I}$ consisting of just the Pauli matrix $\alpha$. If $P$ and $P'$ are two Pauli strings that only consist of Pauli matrices of the same type, then $P$ and $P'$ commute. Thus, $P^X, P^Y, P^Z$ are three disjoint groups of Pauli strings where all the Pauli strings within each group commutes. Identifying these 3 groups as color classes of some coloring, we have that $C(H)$ can be 3-colored, i.e., $\chi(C(H)) \leq 3$.
\end{proof}

\subsection{Handcrafted Colorings for 1D Hamiltonians}
Consider any 1D Hamiltonian $H$ of the form:
\begin{equation}
\label{eqn:general_1D_Ham}
\sum_{j=1}^k d_j X_j + \sum_{j=1}^{k-1}\sum_{\alpha \in \{X,Y,Z\}} c^\alpha_j \alpha_j \alpha_{j+1},
\end{equation}

where $c_j^\alpha$ and $d_j$ are real-valued constants for all $\alpha \in \{X,Y,Z\}$. By Theorem \ref{thm:xyz_coloring}, it is clear that $\chi(C(H)) \leq 3$ by using the XYZ-coloring. We show an alternate coloring scheme for such Hamiltonians that still uses 3 colors but with the additional property that Pauli terms along the same ``edge" in the interaction graph are always grouped together. To see how to do this, we first consider the case where $d_j = 0$ for all $j$.

\begin{theorem}
\label{thm:handcrafted_coloring}
    Let $H$ be a Hamiltonian of the form in Equation \ref{eqn:general_1D_Ham} with $d_j = 0$ for all $j \in [k]$. Then $\chi(C(H)) \leq 2$.
\end{theorem}
\begin{proof}
    Consider any Pauli string of the form $c^\alpha_j \alpha_j \alpha_{j+1}$ where $\alpha \in \{X,Y,Z\}$. If $j$ is odd, then color the Pauli string red, otherwise, if $j$ is even, color it blue. It remains to show that this 2-coloring is proper. Observe that with this coloring, any two Pauli strings with the same color must either be on disjoint qubits (in which case they clearly commute) or the same set of qubits. Suppose, it's the latter case, i.e., we have two Pauli strings $P = c^\alpha_j \alpha_j\alpha_{j+1}$ and $Q = c^\beta_j \beta_j \beta_{j+1}$ with $\alpha, \beta \in \{X,Y,Z\}$ with $\alpha \neq \beta$. Note that $\alpha$ and $\beta$ anticommute, and thus, $P$ and $Q$ have an even number of anticommuting positions, implying that $P$ and $Q$ commute and that the 2-coloring is proper.
\end{proof}

In the case that $H$ has terms $d_jX_j$ with $d_j \neq 0$, the coloring in Theorem \ref{thm:handcrafted_coloring} can be extended by introducing another color (e.g. green) and coloring all the single-site terms with this new color, thus making $\chi(C(H)) \leq 3$.

\section{Results and Evaluation}
\label{sec:results}

\subsection{Experimental Setup}
\label{sec:setup}

\subsubsection{Hamiltonians}
\label{sec:setup_hamiltonians}

We evaluate all ordering strategies on three families of Heisenberg-style Hamiltonians spanning 1D and 2D geometries. The primary system is the transverse-field XXZ spin chain modeling the quantum magnet Cs$_2$CoCl$_4$~\cite{laurell2021},
\begin{equation}
  H_{\mathrm{1D}} = \sum_{\langle i,j \rangle}
    \Bigl( X_i X_j + Y_i Y_j + \Delta\, Z_i Z_j \Bigr)
    + g \sum_i X_i,
  \label{eq:H1D}
\end{equation}
\begin{figure*}[t!]
    \centering
    \includegraphics[width=\textwidth]{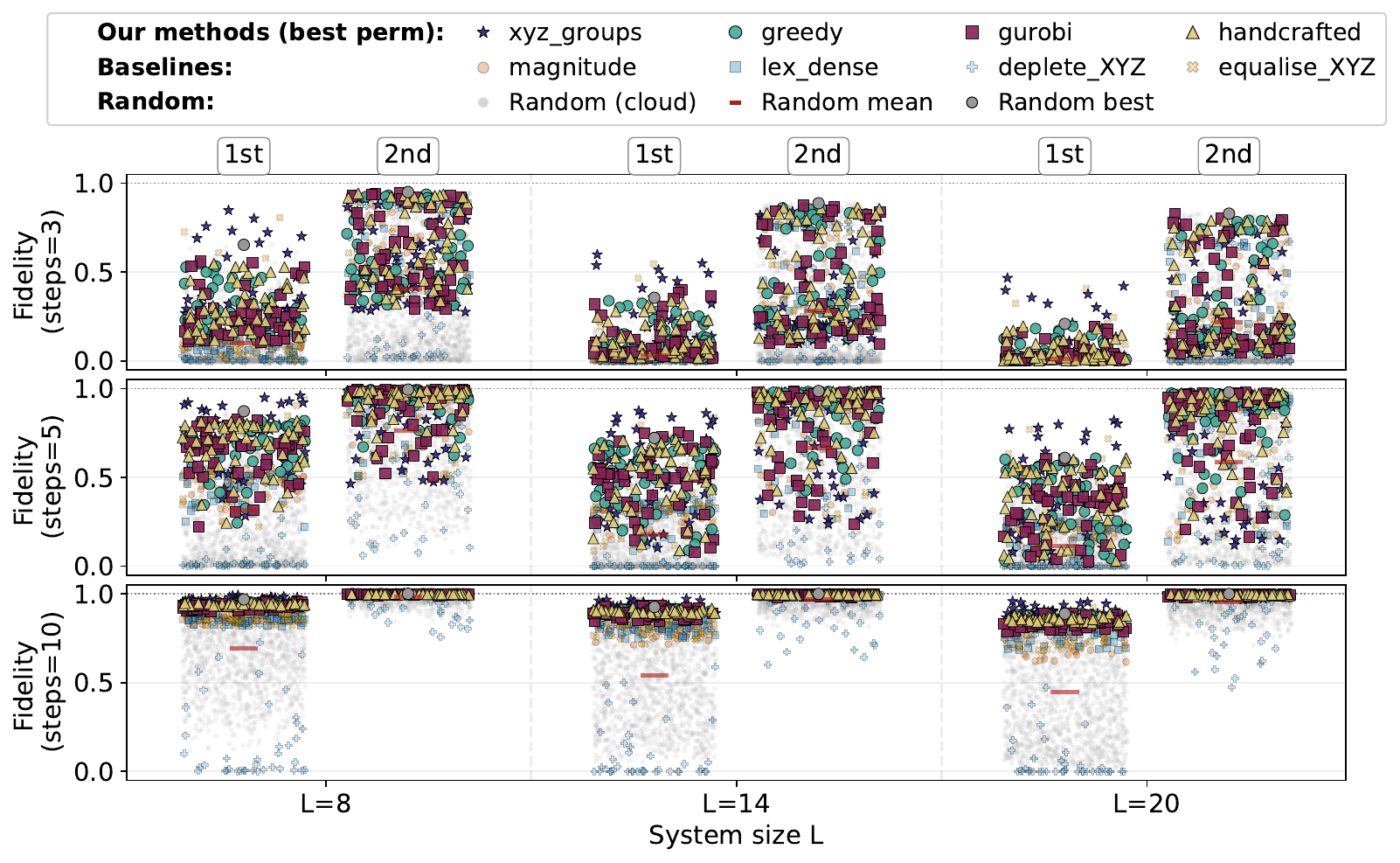}
    \caption{
    Fidelity of each ordering method as a function of system size for the 1D XXZ chain
    with $T = 5.0$. System sizes shown are $L \in \{8, 14, 20\}$, each containing 50 Hamiltonian instances swept over $\Delta$ and $g$. Within each system size group, left and right clusters show 1\textsuperscript{st} and 2\textsuperscript{nd} order Trotter, respectively.
    Rows correspond to Trotter step counts $s \in \{3, 5, 10\}$.
    At low step counts the random fidelity distribution spans nearly the full $[0, 1]$ interval, confirming that ordering choice has a large effect. The commutation graph-based orderings consistently sit in the upper tail of the distribution across all system sizes. As the step count increases to 10, all methods compress toward high fidelity, but the structured orderings still maintain an edge over the random mean.
}
    \label{fig:scatter_1d}
\end{figure*}
expressed in dimensionless units ($J = 1$, $\hbar = 1$), where $\Delta$ is the exchange anisotropy and $g = h_x/J$ is the dimensionless transverse field. We sweep $\Delta \in \{0.12,\, 0.25\}$ and $g \in \{0.1, 0.2, \ldots, 2.5\}$ (25 values). System sizes range from $L = 3$ to $L = 20$ under open boundary conditions. Each instance contains $4L - 3$ Pauli terms: $3(L-1)$ two-body terms of the form $\alpha_i \alpha_{i+1}$ for $\alpha \in \{X, Y, Z\}$, plus $L$ single-site terms $X_i$. This yields 50 Hamiltonians per system size and 900 instances in total. The secondary systems are two families of 2D spin models. The triangular model,
\begin{align}
  H_{\mathrm{tri}} &=
    \sum_{\langle i,j\rangle}
      \bigl(X_iX_j + Y_iY_j + Z_iZ_j\bigr) \nonumber \\
    &\quad + \alpha \sum_{\langle\langle i,j\rangle\rangle}
      \bigl(X_iX_j + Y_iY_j + Z_iZ_j\bigr),
  \label{eq:Htri}
\end{align}
is the $J_1$--$J_2$ Heisenberg antiferromagnet on a triangular lattice, relevant to the quantum spin liquid candidate KYbSe$_2$~\cite{scheie2021witnessing}. The frustration parameter $\alpha = J_2/J_1$ is swept across 27 values in $[0.0,\, 0.50]$. The rectangular model,
\begin{equation}
  H_{\mathrm{rect}} =
    \sum_{\langle i,j\rangle}
      \bigl(X_iX_j + Y_iY_j + Z_iZ_j\bigr)
    + h_x \sum_i X_i,
  \label{eq:Hrect}
\end{equation}
is the 2D analog of $H_{\mathrm{1D}}$, with transverse field $h_x$ swept across 11 values in $[0.0,\, 3.0]$. Both 2D models use open boundary conditions and snake qubit indexing on grids $L_x \times L_y$ with $L_x \in \{2,3,4,5\}$, $L_y \in \{L_x, L_x+1\}$, and $L_x \times L_y \leq 20$, giving 4 to 20 qubits ($L = 4$ to $20$). Further details on lattice geometries and qubit indexing are provided in Appendix~\ref{sec:rectLattice}, Appendix~\ref{sec:triLattice}, and Appendix~\ref{sec:qubitIndexing}. All three Hamiltonian families are non-mixed (each Pauli string contains at most one Pauli type), so the XYZ-coloring from Theorem~\ref{thm:xyz_coloring} applies to all instances.

\begin{figure*}[t]
    \centering
    \includegraphics[width=\linewidth]{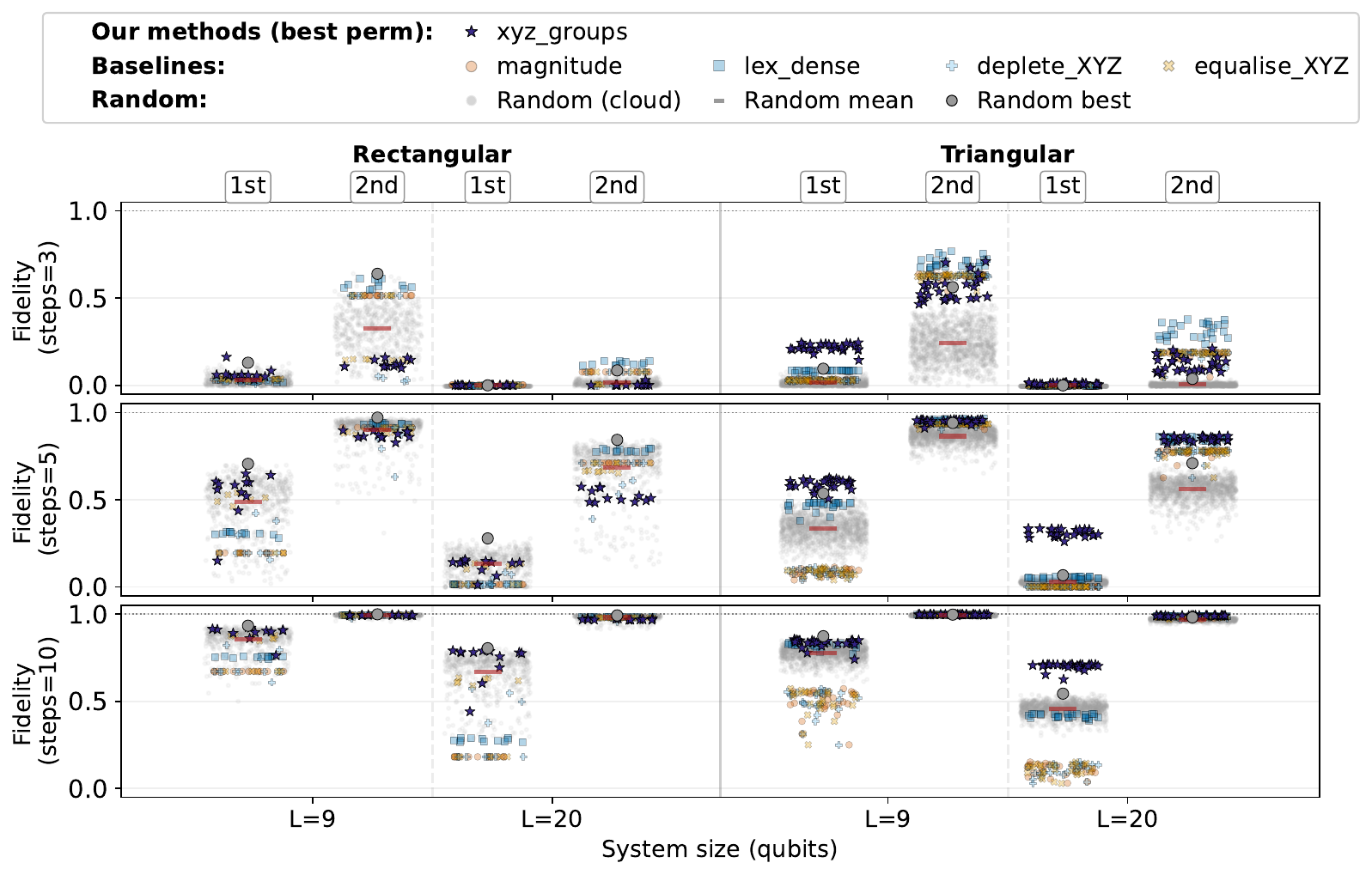}
    \caption{
    Fidelity of each ordering method for 2D rectangular (left) and triangular (right) lattices with $T = 1.0$ and $L \in \{9, 20\}$, each containing 12 rectangular (swept over $h_x$) and 27 triangular (swept over $\alpha$) Hamiltonian instances. Left and right clusters show 1\textsuperscript{st} and 2\textsuperscript{nd} order Trotter; rows show Trotter steps $s \in \{3, 5, 10\}$.
    For the \textbf{rectangular} lattice, at $s=3$ all methods yield near-zero fidelity under 1\textsuperscript{st} order, while lex\_dense achieves the highest fidelity under 2\textsuperscript{nd} order. At $s=5$ and $s=10$, xyz\_groups clearly leads under 1\textsuperscript{st} order, maintaining a gap over all other methods.
    For the \textbf{triangular} lattice, xyz\_groups leads under 1\textsuperscript{st} order across all steps. Under 2\textsuperscript{nd} order, lex\_dense dominates at $s=3$; at $s=5$ and $s=10$ all methods converge to near unity.
    }
    \label{fig:scatter_2d}
\end{figure*}

\subsubsection{Orderings}
\label{sec:setup_orderings}
We evaluate two classes of orderings: commutation-based orderings that exploit the structure of the commutation graph, and baseline orderings from~\cite{tranter2019ordering}.

\textbf{Commutation-based orderings.} For the 1D chain we obtain four distinct 3-colorings of the commutation graph: the XYZ-coloring (Theorem~\ref{thm:xyz_coloring}), a Gurobi-based optimal coloring, a greedy heuristic, and the handcrafted coloring (Theorem~\ref{thm:handcrafted_coloring}). All four yield exactly $m = 3$ commuting subgroups. We evaluate all $m! = 3! = 6$ group permutations for each, giving 24 \om orderings per Hamiltonian. For the 2D models we evaluate the XYZ-coloring with all $m! = 6$ group permutations. We also evaluate the depleteGroups and equaliseGroups strategies from~\cite{tranter2019ordering}, which build a flat ordering by picking terms one at a time from the XYZ groups based on coefficient magnitude.

\begin{figure*}[t!]
    \centering
    \begin{subfigure}[b]{\linewidth}
        \centering
        \includegraphics[width=\linewidth]{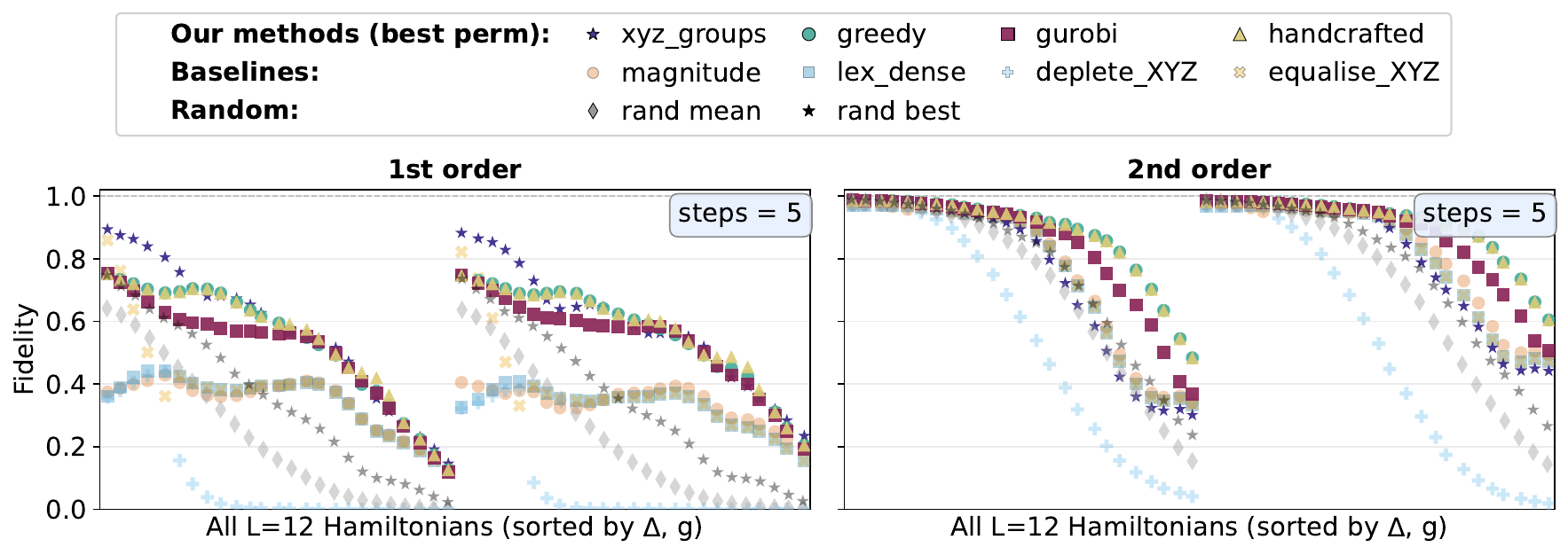}
        \caption{1D XXZ chain at $L=12$ with 50 Hamiltonians sorted by $\Delta$ and $g$.}
        \label{fig:slice_1d}
    \end{subfigure}
    
    \vspace{0.5em}
    
    \begin{subfigure}[b]{\linewidth}
        \centering
        \includegraphics[width=\linewidth]{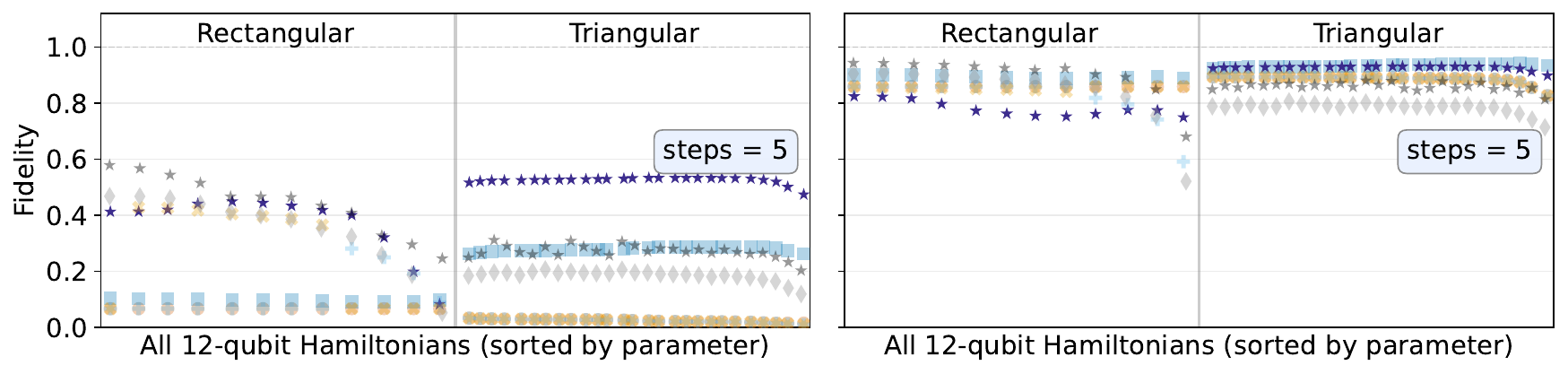}
        \caption{2D rectangular (12 Hamiltonians) and triangular (27 Hamiltonians) lattices at $L = 12$.}
        \label{fig:slice_2d}
    \end{subfigure}
\caption{
    Per-Hamiltonian fidelity of all ordering methods at $s=5$ with $L=12$.
    (\textit{Top-left: }) For the 1D chain under 1\textsuperscript{st} order Trotter, xyz\_groups consistently achieves the highest fidelity, while greedy, gurobi, and handcrafted cluster together well above the baselines.
    (\textit{Top-right: }) Under 2\textsuperscript{nd} order, greedy, gurobi, and handcrafted reach the top, with xyz\_groups falling slightly below at higher $g$.
    (\textit{Bottom-left: }) For the 2D lattices under 1\textsuperscript{st} order, xyz\_groups clearly dominates for both rectangular and triangular lattices, with all other methods well below.
   (\textit{Bottom-right: }) Under 2\textsuperscript{nd} order, for the rectangular lattice random best and lex\_dense lead and for the triangular lattice, xyz\_groups dominate at top.
    }
    \label{fig:per_hamiltonian_slices}
\end{figure*}

\textbf{Baseline orderings.}
We evaluate the two non-commutation-based orderings described in Section \ref{sec:baseline_orderings}.
%We evaluate four non-commutation-based orderings: magnitude ordering (mag\_ties), which sorts terms by descending coefficient magnitude with ties grouped together; dense lexicographic ordering (lex\_dense), which sorts by the full Pauli string; sparse lexicographic ordering (lex\_sparse), which sorts first by the active qubit indices and then by Pauli type; and magnitude-then-sparse-lex ordering (mag\_sparse\_lex), which sorts by descending coefficient magnitude and breaks ties using the sparse lexicographic order. For mag\_ties, which produces groups of terms with equal coefficient magnitude, we evaluate all permutations of these groups.

\textbf{Random baseline.} As a stochastic reference we sample $N = 30$ uniformly random term-level permutations per instance for the 1D and 2D models. We report the mean and best fidelity over these random orderings.

\subsubsection{Evaluation Methodology}
\label{sec:eval_method}

For each ordering we compute the state-dependent fidelity between the Trotterized and exact time-evolved states as defined in Eq.~\eqref{eq:fidelity}. The exact state $|\psi_{\mathrm{exact}}\rangle = e^{-iHT}|\psi_0\rangle$ is obtained via sparse Krylov matrix exponentiation, which avoids constructing the full dense matrix exponential. Writing $H = \sum_{j=1}^k c_j P_j$ where $c_j$ are real coefficients and $P_j$ are Pauli strings, the Trotterized state is computed by sequentially applying precomputed per-term unitaries $e^{-i c_j P_j \delta t}$, where $\delta t = T/s$ is the time step and $s$ is the number of Trotter steps. Since each Pauli string satisfies $P_j^2 = I$, the matrix exponential reduces to
\begin{equation}
  e^{-i c_j P_j \delta t} = \cos(c_j \delta t)\, I - i \sin(c_j \delta t)\, P_j,
  \label{eq:pauli_exp}
\end{equation}
which can be applied as a sparse matrix-vector product in $\mathcal{O}(2^L)$ time per term. All per-term unitaries are precomputed once before iterating over orderings, so the cost of evaluating each ordering is $\mathcal{O}(k \cdot s \cdot 2^L)$ where $k$ is the number of Pauli terms. % and $s$ is the number of steps.

We evaluate both 1\textsuperscript{st} order and 2\textsuperscript{nd} order (Suzuki) Trotter decompositions at step counts $s \in \{3, 5, 10, 20\}$. For 1\textsuperscript{st} order Trotter, each step applies the ordering once. For 2\textsuperscript{nd} order Trotter, each step applies the ordering followed by the reverse ordering, with each term evolved for half the time step $\delta t / 2$. The total evolution time is $T = 5.0$ for the 1D chain and $T = 1.0$ for the 2D models. The shorter time for 2D reflects the larger operator norm due to higher connectivity.

The ordering strategies proposed in this work do not depend on the choice of initial state, as they depend only on the commutation structure of the Hamiltonian. Any arbitrary initial state could be used for evaluation, and we choose the N\'{e}el state $|0101\ldots\rangle$~\cite{auerbach1994} as it produces non-trivial dynamics for all Hamiltonians considered. We evaluate all $m!$ group permutations and report the best-performing permutation, as indicated by ``best perm'' in the figure legends. Across all figures except Figure~\ref{fig:summary}, the legend is organized into three rows: the first row shows our commutation-based methods reporting the best permutation over all $m!$ group orderings, the second row shows the baseline methods, and the third row shows the random reference consisting of a grey cloud of individual random orderings, their mean, and their best. Figure~\ref{fig:summary} additionally includes a permutations row that shows all individual group orderings for each commutation-based method.

% \subsubsection{Orderings Evaluation}
% \tate{I talked a bit about Trotter error in the background section, will need to make some adjustments here to avoid repeating ourselves too much.}
% To evaluate the orderings generated from three different graph coloring techniques, we compute simulation fidelity which is measured as
% %
% \begin{equation}
%   F = \bigl|\langle \psi_{\mathrm{exact}} \mid
%              \psi_{\mathrm{Trotter}} \rangle\bigr|^2,
%   \label{eq:fidelity}
% \end{equation}
% %
% where $|\psi_{\mathrm{exact}}\rangle = e^{-iHT}|\mathbf{0}\rangle$ is obtained by sparse Krylov matrix exponentiation and $|\psi_{\mathrm{Trotter}}\rangle$ is the Trotterized state. We evaluate both first-order and second-order (Suzuki) decompositions at Trotter step counts $r \in \{3,\, 5,\, 10,\, 20,\, 40\}$ with total evolution time $T = 5.0$ (dimensionless). All per-term unitaries are precomputed once before iterating over orderings, reducing per-ordering cost from $\mathcal{O}(8^L)$ to $\mathcal{O}(2^L)$.

\begin{figure*}[t!]
    \centering
    \includegraphics[width=\linewidth]{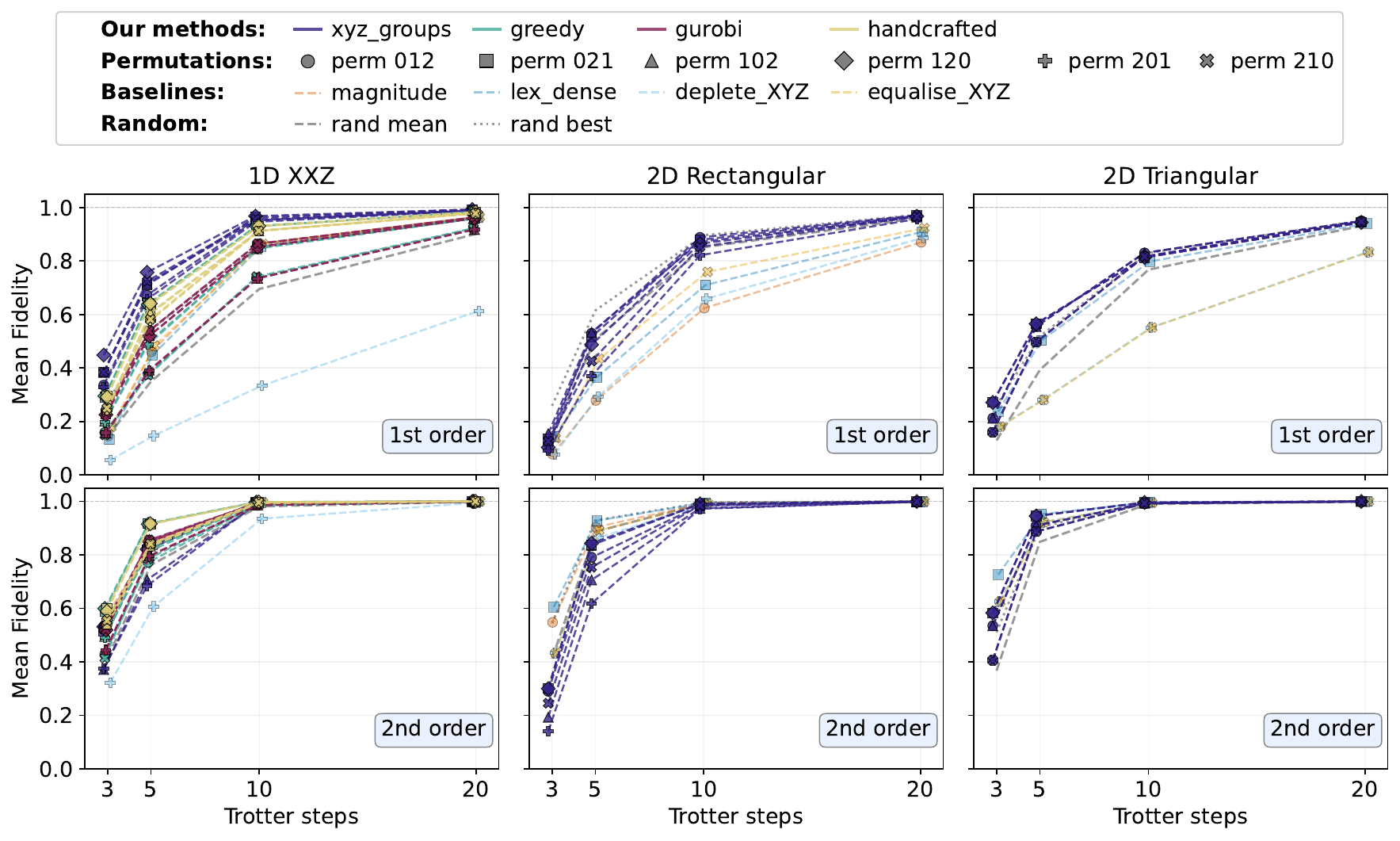}
    \caption{
    Mean fidelity per ordering method as a function of Trotter steps, averaged over all Hamiltonians and system sizes. Columns show 1D XXZ ($T=5.0$, $L=3$ to $20$), 2D rectangular, and 2D triangular lattices ($T=1.0$, $L=4$ to $20$). Rows show 1\textsuperscript{st} and 2\textsuperscript{nd} order Trotter. Lines show all permutations of our commutation graph-based orderings and baseline heuristics~\cite{tranter2019ordering}.
    Under 1\textsuperscript{st} order Trotter, xyz\_groups achieves the highest mean fidelity across all three datasets, with all methods converging near unity by $s=20$.
    Under 2\textsuperscript{nd} order Trotter, all our methods perform competitively for the 1D chain, while convergence to unity is rapid for both 2D lattices.
   Across conditions, perm~120 ($\diamond$) is the dominant permutation for xyz\_groups under 1D 1\textsuperscript{st} order and for handcrafted under 1D 2\textsuperscript{nd} order. For 2D rectangular 2\textsuperscript{nd} order, xyz\_groups perm~120 ($\diamond$) dominates. For 2D triangular 2\textsuperscript{nd} order, xyz\_groups perm~021 ($\square$) wins across all step counts.
    }
    \label{fig:summary}
\end{figure*}

\subsection{Random vs. Structured Orderings}
\label{sec:ordering_matters}

We first establish that term ordering has a practically significant effect on Trotter fidelity. Figures~\ref{fig:scatter_1d} and~\ref{fig:scatter_2d} show the fidelity of every ordering method for individual Hamiltonian instances across multiple system sizes and step counts, for the 1D XXZ chain ($L \in \{8, 14, 20\}$) and 2D lattices ($L \in \{9, 20\}$) respectively. The grey cloud provides a reference distribution against which all structured methods can be compared. 

For the 1D chain (Figure~\ref{fig:scatter_1d}), the random fidelity distribution is strikingly wide. At steps$\,=3$ and steps$\,=5$ under 1\textsuperscript{st} order Trotter, the fidelities of 30 random orderings span nearly the full $[0, 1]$ interval at every system size. Many random orderings produce near-zero fidelity while a few approach unity. This spread confirms that ordering is not a minor perturbative effect. A poor choice of ordering can severely degrade fidelity, even at moderate system sizes and step counts. The commutation-based orderings consistently sit in the upper tail of the random distribution across all system sizes. At steps$\,=3$ and steps$\,=5$, the best structured orderings achieve fidelities well above the random mean at every system size. As the step count increases to 10, all methods compress toward high fidelity, but the structured orderings still maintain an edge over the random mean. Under 2\textsuperscript{nd} order Trotter, fidelity improves across all orderings, but the relative advantage of structured methods over random ones persists.

\begin{figure*}[t!]
    \centering
    \includegraphics[width=\linewidth]{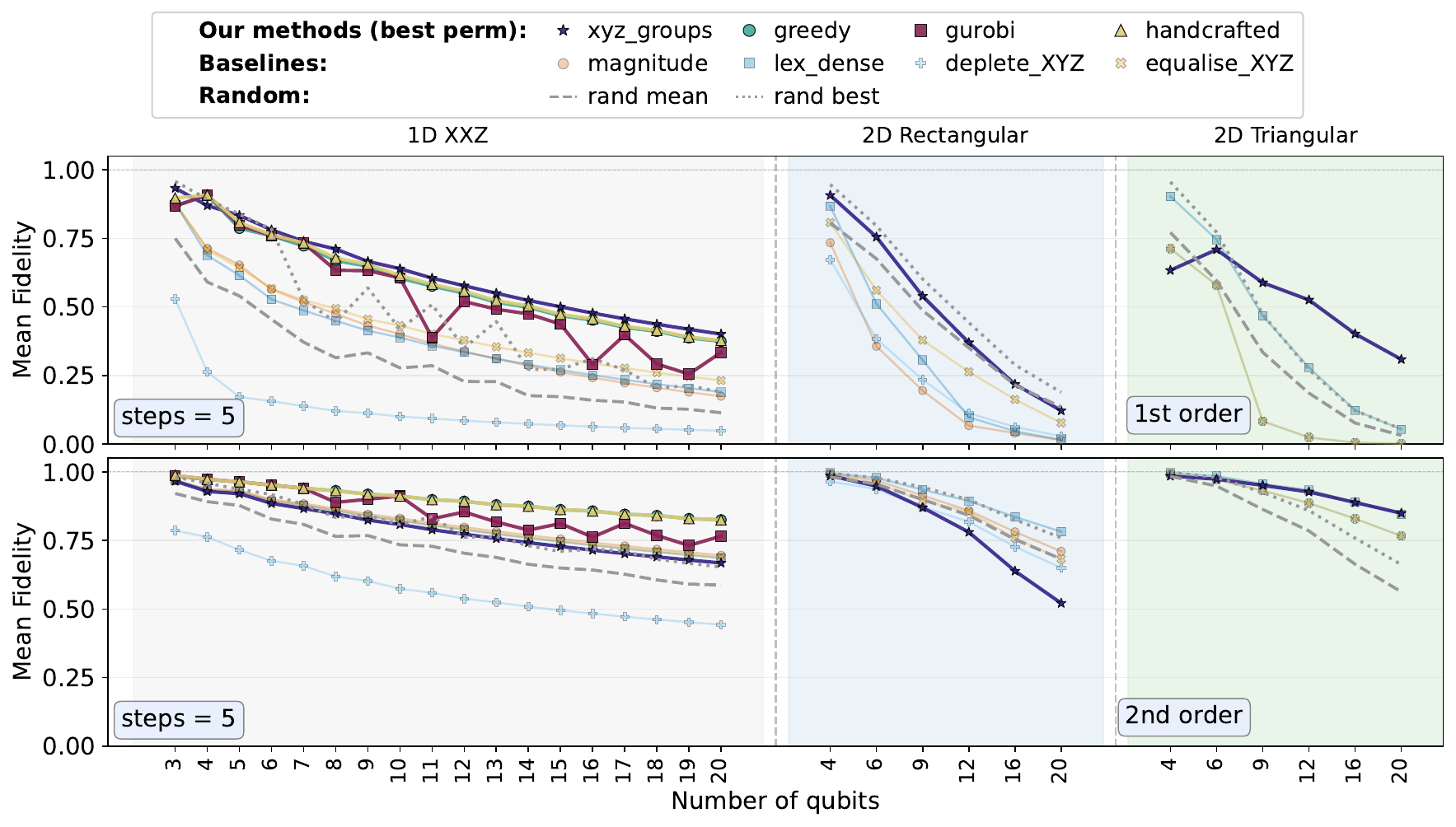}
    \caption{
    Mean fidelity as a function of system size at $s=5$ ($T=5.0$ for 1D, $T=1.0$ for 2D). Columns show 1D XXZ ($L=3$ to $20$), 2D rectangular, and 2D triangular ($L=4$ to $20$); rows show 1\textsuperscript{st} and 2\textsuperscript{nd} order Trotter. Thick and thin lines show our commutation graph-based orderings and baselines~\cite{tranter2019ordering}.
    Under 1\textsuperscript{st} order, xyz\_groups and other coloring orders cluster at the top in 1D, while xyz\_groups clearly leads in both 2D lattices with fidelity declining steeply with system size.
    Under 2\textsuperscript{nd} order, greedy, gurobi, and handcrafted lead in 1D with xyz\_groups falling below them. For the rectangular lattice, xyz\_groups degrades more steeply than other methods. For the triangular lattice, xyz\_groups leads as in 1\textsuperscript{st} order but with a less steep decline.
    }
    
    \label{fig:scaling}
\end{figure*}

The ordering effect extends to 2D systems as well. Figure~\ref{fig:scatter_2d} shows fidelity distributions for rectangular and triangular lattices at system sizes $L = 9$ and $L = 20$. The random distribution is again wide at low step counts, and structured orderings again sit in the upper tail. The effect is more pronounced in 2D than in 1D. At $L = 20$ with steps$\,=5$ under 1\textsuperscript{st} order Trotter, most random orderings produce fidelities below 0.1, while the best structured ordering (xyz\_groups) achieves substantially higher values. At steps$\,=3$ under 2\textsuperscript{nd} order Trotter, the lex\_dense ordering performs notably well across both rectangular and triangular lattices. The 2D lattices have more interaction terms per qubit due to higher connectivity, which amplifies the sensitivity to ordering choice.

These results demonstrate that the choice of Trotter ordering has a first-order impact on simulation fidelity. The effect is strongest at low step counts and larger system sizes, which is precisely the regime of greatest practical interest on near-term quantum hardware where circuit depth is limited. In all subsequent sections we focus on identifying which specific ordering strategies perform best and how their advantage scales with system size.

\subsection{Comparison of Ordering Methods}
\label{sec:coloring_comparison}
We now compare the fidelities achieved by each ordering strategy. Figure~\ref{fig:per_hamiltonian_slices} shows per-Hamiltonian fidelity at $s=5$ and $L = 12$.

\subsubsection{1D XXZ Chain}
For the 1D chain at $L = 12$ (Figure~\ref{fig:slice_1d}), xyz\_groups generally achieves the highest fidelity under 1\textsuperscript{st} order Trotter, though other graph-coloring-based orderings (greedy, gurobi, handcrafted) match xyz\_groups at larger values of $g$. These methods form a cluster well above the random baseline, clearly separated from the remaining methods. The magnitude and lexicographic baselines fall closer to the random mean. The depleteGroups and equaliseGroups strategies from~\cite{tranter2019ordering} perform noticeably worse than the \om approach.

Under 2\textsuperscript{nd} order Trotter (Figure~\ref{fig:slice_1d} right), all methods achieve higher fidelity than under 1\textsuperscript{st} order and the separation between them shrinks. Greedy and handcrafted dominate overall, with xyz\_groups falling slightly behind at larger values of $g$. The baseline orderings follow the structured orderings but with a consistent gap.

\subsubsection{2D Lattices}
The advantage of xyz\_groups widens further in 2D. At $L = 12$ (Figure~\ref{fig:slice_2d}) under 1\textsuperscript{st} order Trotter at $s=5$, xyz\_groups clearly leads both lattice types, reaching fidelities around 0.4 to 0.5 for the rectangular lattice and around 0.5 for the triangular lattice, while most other methods remain below 0.1 for the rectangular lattice and around 0.2 to 0.3 for the triangular lattice.

Under 2\textsuperscript{nd} order Trotter in 2D, xyz\_groups falls below the top for the rectangular lattice but dominates for the triangular lattice. Among the baselines, lex\_dense performs well for both lattice types. 

\subsection{Permutation Analysis}

Figure~\ref{fig:summary} shows mean fidelity across all Hamiltonians and system sizes as a function of step count, with all $m!$ group permutations shown explicitly, where a label such as $120$ denotes applying subgroups in the order $(G_1, G_2, G_0)$. 
% The aggregated view confirms the per-Hamiltonian observations and reveals consistent trends in which group permutation performs best across datasets and Trotter orders.

For the 1D chain under 1\textsuperscript{st} order Trotter (Figure~\ref{fig:summary} top-left), xyz\_groups achieves the highest mean fidelity at every step count from $s = 3$ to $s = 20$, with all methods converging toward unity at $s \geq 20$. While all permutations remain competitive, perm~120 ($\diamond$) emerges as the dominant choice for xyz\_groups, winning over $77\%$ of Hamiltonians across all step counts. Under 1D 2\textsuperscript{nd} order (Figure~\ref{fig:summary} top-center), greedy and handcrafted lead overall, with perm~012 ($\circ$) taking over as the leading permutation for xyz\_groups ($53$--$70\%$), perm~021 ($\square$) dominating for greedy ($58$--$62\%$), and perm~120 ($\diamond$) remaining dominant for handcrafted ($48$--$50\%$).

For the 2D rectangular lattice (Figure~\ref{fig:summary} bottom-center), perm~012 leads under 1\textsuperscript{st} order and perm~120 under 2\textsuperscript{nd} order. For the 2D triangular lattice (Figure~\ref{fig:summary} bottom-right), perm~021 ($\square$) is remarkably stable under 2\textsuperscript{nd} order, winning nearly $100\%$ of Hamiltonians across all step counts, while all methods converge to unity by $s = 10$. These trends suggest that the optimal permutation depends on both the Trotter order and the lattice geometry, and that a modest permutation search can reliably identify a strong ordering.

\subsection{Scaling Behavior}
\label{sec:scaling}
We now examine how structured orderings scale with system size. Figure~\ref{fig:scaling} shows mean fidelity as a function of system size at $s=5$ for 1D and 2D systems across three panels.

For the 1D chain under 1\textsuperscript{st} order Trotter (Figure~\ref{fig:scaling}, left panel), fidelity decreases steadily with system size for all methods. This is expected because larger systems have more Pauli terms and thus more non-commuting pairs, increasing the Trotter error at fixed step count. The xyz\_groups ordering degrades the slowest, maintaining the highest mean fidelity from $L = 3$ through $L = 20$, with greedy, gurobi, and handcrafted tracking closely behind. The baseline methods and random orderings decay faster, and the gap between xyz\_groups and the random mean widens as $L$ increases. Under 2\textsuperscript{nd} order Trotter, fidelity is higher across all system sizes and the advantage shifts to greedy, gurobi, and handcrafted, which outperform xyz\_groups at larger system sizes.

For the 2D rectangular lattice (Figure~\ref{fig:scaling}, center panel), the decline is steeper due to higher connectivity and more non-commuting term pairs per qubit. Under 1\textsuperscript{st} order, xyz\_groups leads but declines steeply. Under 2\textsuperscript{nd} order, xyz\_groups degrades more steeply than other methods. For the 2D triangular lattice (Figure~\ref{fig:scaling}, right panel), xyz\_groups leads under both Trotter orders despite a steep decline with system size. Under 2\textsuperscript{nd} order, the decline is less steep than under 1\textsuperscript{st} order, and xyz\_groups maintains its lead across all system sizes.

These results show that the advantage of commutation-based orderings grows with system size under 1\textsuperscript{st} order Trotter. Under 2\textsuperscript{nd} order Trotter, the spatial-locality-based colorings (greedy, gurobi, handcrafted) outperform xyz\_groups in 1D at larger system sizes, suggesting that the optimal grouping strategy depends on both the Trotter order and the lattice geometry.

\section{Discussion}
Our results show that our \om method often performs quite well, especially when the xyz-coloring is used. However, many of the plots show that the relative performance of our methods is heavily dependent on the exact Trotter settings and Hamiltonian, and sometimes in a way that is not always intuitively clear. For example, the lex\_dense method does surprisingly well compared to other methods specifically for the 2D triangular lattice when 2nd-order Trotterization with $s=3$ steps is used. A theoretically-backed explanation for some of these phenomena may allow for the construction of better future ordering strategies.

Although our method, \om, takes the commutation structure of the graph into account, it is entirely agnostic to the coefficient values associated with each Pauli string. It would be interesting to see an adaptation that maintains the spirit of our approach (keeping commuting terms together) in a way that uses the coefficient information.

The systems we consider are all Heisenberg-style systems, meaning that they can be colored with 3 or less colors; in future work, we would like to consider more exotic Hamiltonians with ``mixed" terms, e.g., structured Hamiltonians (with some notion of an interaction graph) and molecular Hamiltonians.

Although it yields the same number of colorings at the XYZ-coloring, the handcrafted coloring strategy for the 1D systems have some nice properties. First, it can be shown that Trotterization with such a scheme will preserve the particular number, an important quantity in the context of molecular Hamiltonians. This scheme could also easily be extended to other systems with more complex interaction graphs (or hypergraphs), and a similar analysis could be performed to obtain bounds on the chromatic number of the corresponding commutation graphs. Additionally, for many Hamiltonians, their commutation graphs have additional structure. Exploring the structural properties of such graphs may lead to better strategies for finding optimal or near-optimal vertex colorings on such graphs.

Lastly, the commutation graph is closely related to the notion of a \emph{compatibility graph}, whose colorings correspond to potential gate parallelizations, thus reducing circuit depth. The correlation and potential tradeoff between Trotter error and circuit depth for different orderings would make for an interesting future study.

% \begin{itemize}
%     \item Mention something regarding coefficient-aware orderings
%     \item Future work: molecular hamiltonians, learning might be better there since number of colors grows with system size
% \end{itemize}

\section{Acknowledgments}
This work was supported by the U.S. Department of Energy through the Los Alamos National Laboratory. Los Alamos National Laboratory is operated by Triad National Security, LLC, for the National Nuclear Security Administration of U.S. Department of Energy (Contract No. 89233218CNA000001). Research presented in this article was supported by the NNSA's Advanced Simulation and Computing Beyond Moore's Law Program at Los Alamos National Laboratory and by the Laboratory Directed Research and Development program of Los Alamos National Laboratory under project number 20230049DR. LANL report LA-UR-26-23269.

\appendices

\section{2D Hamiltonian Details and Indexing Scheme}
\subsection{Rectangular Grid}
\label{sec:rectLattice}
For a $W \times L$ rectangular grid, we consider qubits at sites $(w, \ell)$ with $w \in \{0, 1, \dots, W-1\}$ and $\ell \in \{0, 1, \dots, \ell\}$. Two sites are connected by an edge if their euclidean distance is exactly 1, i.e., each site will have 4 neighbors in the four main cardinal directions (up, down, left, right). We only consider open boundary conditions (i.e. non-periodic).
\subsection{Triangular Lattice}
\label{sec:triLattice}
Every point in the triangular lattice can be considered as a linear combintation of the vectors $v_1 = (1,0)$ and $v_2 = (\frac{1}{2}, \frac{\sqrt{3}}{2})$. Considering all linear combinations $wv_1 + \ell v_2$ with  $w \in \{0, 1, \dots, W-1\}$ and $\ell \in \{0, 1, \dots, L\}$ produces a skewed grid of points in the shape of a right-leaning parallelogram (see Figure \ref{fig:tri_lattice}). In the basis $\mathcal{B} = \{v_1, v_2\}$, the previously mentioned collection of points forms a \emph{rectangular} grid of points of the form  $(w, \ell)$ with $w \in \{0, 1, \dots, W-1\}$ and $\ell \in \{0, 1, \dots, \ell\}$. For each site, the nearest and next-nearest neighbors (with respect to the original coordinate system) are connected by an edge with weight 1 and $\alpha$ respectively. We only consider open boundary conditions (i.e. non-periodic).

\subsection{Qubit Indexing}
\label{sec:qubitIndexing}
In regards to the (linear) qubit indexing, we use the standard indexing scheme for the 1D system (where an intersection occurs between qubit $i$ and qubit $j$ if and only if $|i-j| = 1$). For the 2D Hamiltonians we consider, the qubit positions  (possibly after some change of basis) form a rectangular grid. We choose a linear indexing scheme where the qubits are labeled by following a ``snaking" pattern through the various rows of the grid, starting in the bottom-left corner (see Figure \ref{fig:tri_lattice}). This ensures that qubits with close indices correspond to qubits that are physically close together (in the grid), note that the converse is not necessarily true.

\begin{figure}
\includegraphics[width=\linewidth]{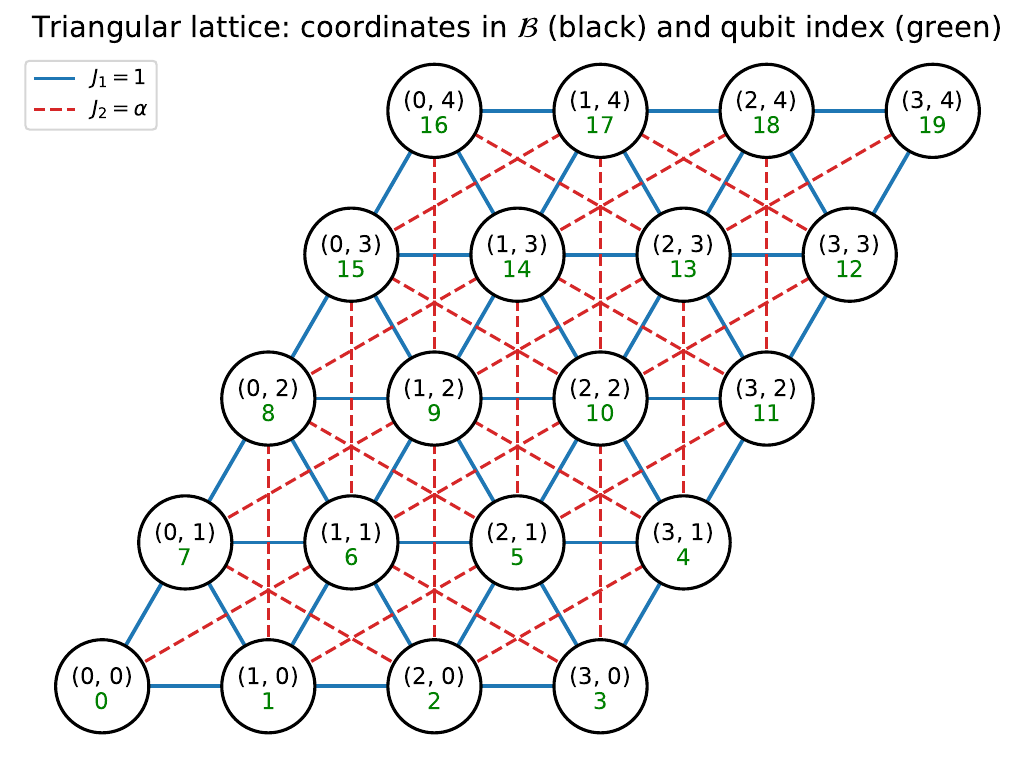}
\caption{The interaction graph for the 2D triangular lattice Hamiltonian with width $W = 4$ and $L = 5$. While the positions of the sites form a paralleogram, they form a rectangular grid with respect to the basis $\mathcal{B}$ described in Appendix \ref{sec:triLattice}. The coordinates in $\mathcal{B}$ are given in black and the qubit index (formed using the snaking pattern) is given in green for each site/qubit. The blue (solid) and red (dashed) edges represent the nearest and next-nearest neighbor interactions respectively.}
\label{fig:tri_lattice}
\end{figure}

\clearpage

\bibliographystyle{unsrt}
\bibliography{sample}
\end{document}